\documentclass{article}
\usepackage[utf8]{inputenc}
\usepackage{amsthm}
\usepackage{amsmath}
\usepackage{algorithmic}
\usepackage{graphicx}
\usepackage{subfigure}
\usepackage{array}
\usepackage[babel=true]{csquotes}
\newtheorem{Lemma}{Lemma}
\newtheorem{Theorem}{Theorem}
\theoremstyle{plain}

\begin{document}

\title{Resiliency in Distributed Sensor Networks for PHM of the Monitoring Targets}
\author{J. Bahi$^{1}$, W. Elghazel$^{2}$, C. Guyeux$^{1}$, M. Haddad$^{3}$, M. Hakem$^{1}$,\\ K. Medjaher$^{2}$, and N. Zerhouni$^{2}$}


\maketitle

\begin{abstract}
In condition-based maintenance, real-time observations are crucial for on-line health assessment. When the monitoring system is a wireless sensor network, data loss becomes highly probable and this affects the quality of the remaining useful life prediction. In this paper, we present a fully distributed algorithm that ensures fault tolerance and recovers data loss in wireless sensor networks. We first theoretically analyze the algorithm and give correctness proofs, then provide simulation results and show that the algorithm is (i) able to ensure data recovery with a low failure rate and (ii) preserves the overall energy for dense networks.
\end{abstract}

\section{Introduction}

In a monitoring activity, the sensor nodes are placed on/around the monitored target to collect measurements of relevant parameters, such as temperature. These measurements will help evaluate the system's current state of health, diagnose the degree of its severity, and extrapolate the result in the future to estimate when the system is more likely to fail. The goal from this activity is to schedule maintenance in a way that prevents system failure and shutdown. To guarantee the efficiency of this process, the accuracy of on-line measurements is a crucial requirement. Consequently, the Wireless Sensor Network (WSN) used in the monitoring needs to be dependable. 

Avizienis \cite{Avizienis00} defined system dependability as \enquote{the ability of a system to avoid failures that are more frequent or more severe, and outage durations that are longer, than is acceptable to the users}. A dependable network should be able to deliver a correct service (forwards measurements to the base station) and makes sure that failed components will not lead to a network failure. Dependability of WSNs is a property that integrates the attributes needed for the application to be justifiably trusted. These attributes include availability and reliability. 

A network failure can be caused by a number of triggers such as: packet loss, node failure, energy exhaustion, packet interference... The network is considered available if its downtime is very limited, either due to few failures or to quick restarts when a failure takes place \cite{Knight04,Taherkordi06}. A reliable network is able to continuously deliver a correct service. The reliability can be computed as the probability that a network functions properly during a time interval \cite{Taherkordi06,Silva12}. 

Most of the research works solved the problem of network reliability through retransmission and redundancy mechanisms \cite{Silva12}.\\The acknowledgment mechanism is employed by the receiver to notify the sender of the reception status. If the packet fails to arrive to its destination, the sender keeps on re-sending it until the transmission is successful \cite{Akan05,Zhou05,Gungor06,Iyer05}. Unfortunately, this solution does not respect the energy constraints of WSNs, since packet transfer consumes the highest amount of energy in the network. Reliability can also be introduced via data redundancy mechanisms. A packet is transmitted in multiple copies using different routes as a backup plan in case one of the routes fails \cite{Al-Wakeel07,Dijkstra74,Mojoodi11}. However, this solution results in unnecessary transmissions and therefore does not improve energy consumption WSNs. 

In the context of of extending the network's lifetime, a possible solution is to maintain a minimum number of sensor nodes in an active mode \cite{HeCLSS12,HeCYSS12,KasbekarBS11}. Although this seems to solve the energy problem, other issues arise:

\begin{itemize}

\item How can we ensure a minimum coverage rate?
\item How can we reduce the loss of data?
\item How can we avoid unnecessary packet forwards?

\end{itemize}

In this study, the number of awake sensor nodes is kept to a minimum; enough to ensure coverage rate. The probability of nodes awakening is updated following two variables: time and failure rate. Data sensed by a sensor is copied on its neighbors, and will only be retrieved when the active node has failed. this mechanism avoids unnecessary packet forwards and therefore preserves the overall energy. The remainder of this paper is organized as follows. Section \ref{work} presents some of the existing work in WSN reliability. In Section \ref{solution}, we describe our algorithm. The simulation results are shown and discussed in Section \ref{results}.

\section{Related work}
\label{work}

Reliability is an important attribute for WSNs dependability, and it means that the network should be able to continuously deliver a correct service. In order to attain reliability in WSNs, sensing coverage and sensing level need to be considered. The sensing coverage refers to the integrated sensing area which is monitored by at least one sensor node. As for the sensing level, it refers to the number of sensor nodes being able to detect a new event when it takes place \cite{Choi09}. Choi \textit{et al.} argue that existing node scheduling schemes focus on the minimum sensing level for the coverage problem and neglect the fault tolerance issues \cite{Choi09}. In one hand, the minimum sensing level is an NP-complete problem. On the other hand, it cannot be preserved when nodes start to fail. Therefore, the authors propose the Fault-tolerant Adaptive Node Scheduling (FANS) algorithm, which efficiently handles the degradation of the sensing level. The algorithm designates a set of backup nodes for each active node. If the latter fails, the predesignated set of backup nodes activate themselves to replace it and to restore the lowered sensing level. FANS requires a small number of backup nodes and a small amount of control messages. In \cite{Chen12}, Chen \textit{et al.} study fault tolerant out-of-band monitoring for WSNs. They aim at placing a minimum number of monitors in a sensor network in a way that all sensor nodes are monitored by $k$ distinct monitors, and each monitor serves at most $w$ sensor nodes. The authors first prove that this problem is NP-hard and then propose three algorithms providing near optimal solutions. 

Battery level, broken links, and communication failures have an impact on the Quality of Service (QoS) of WSNs. This leads to consequences varying from disturbing the traffic in the network to completely interrupting it. Geeta \textit{et al.} \cite{Geeta13} propose an Active node-based Fault Tolerance using Battery power and Interference model (AFTBI) to identify the faulty nodes in WSNs. Fault tolerance against low battery power is assured through a hand-off mechanism where the faulty node selects the neighbor with highest battery level and transfers all the services towards it. To reduce interference signal, a dynamic power level mechanism is introduced, where the power of a node is adjusted automatically with regards to its current state (active or asleep). Simultaneous transmissions can be avoided if the nodes are only allowed to transmit data within a time slot. Lee and Choi \cite{Lee08} tackle the same problem by identifying and isolating the faulty sensor nodes in the network. Sensed data is compared among neighbors to determine its accuracy. Once the predetermined fault threshold is reached, the node in question is isolated from the diagnosis process; a faulty node can be included in data transferring but not data sensing. Transient faults in communication and sensor reading are tolerated by using the time redundancy mechanism. The drawback of this solution is that faults are assumed to be only related to the sensing activity, excluding other sources of failure.

Energy in WSNs can also be preserved through lifetime optimization. The authors in \cite{Kasbekar11} leverage prediction to prolong the network lifetime, by exploiting temporal-spatial correlations among the data sensed by different sensor nodes. Based on Gaussian Process, the authors formulate the issue as a minimum weight sub-modular set cover problem and propose a centralized and a distributed truncated greedy algorithms (TGA and DTGA). They prove that these algorithms obtain the same set cover. Lifetime optimization using knowledge about the dynamics of stochastic events has been studied in \cite{He12}. The authors presented the interactions between periodic scheduling and coordinated sleep for both synchronous and asynchronous dense static sensor network. They show that the event dynamics can be exploited for significant energy savings, by putting the sensors on a periodic on/off schedule. In \cite{HeCYSS12}, the authors design a polynomial-time distributed algorithm for maximizing the lifetime of the network. They proved that the lifetime attained by their algorithm approximates the maximum possible lifetime within a logarithmic approximation factor. Zhang \textit{et al.} \cite{Zhang14} presented a stochastic sensing algorithm to reduce energy consumption through node scheduling. They used data correlation between nodes to reduce error rate by adjusting duty cycle of faulty sensors. Their algorithm conserves 60 \% of energy as compared to other solutions, while confining sensing error within specified error tolerance. In \cite{He12*}, He \textit{et al.} use actors to allocate spare sensors to sensor-deficient regions or to relocate sensors from sensor-abundant regions to sensor-deficient regions. They introduce a baseline centralized greedy algorithm for sensor allocation, where global sensor information is communicated to obtain the optimal solution. The works cited here focus on a periodic schedule for turning the sensors on and off.

Data collection delay and reliability need to be considered in scheduling algorithms for WSNs. Zhang \textit{et al.} \cite{Zhang12} claim that existing algorithms have not solved these two problems effectively. The authors propose the Fault-Tolerant Scheduling (FTS) algorithm, where each sensor node detects the environment and generates some sensing data at regular intervals. The algorithm helps surviving network malfunction by switching the parent of a sensor node to its backup parent. The simulation results show that FTS has a short data collection time and high fault tolerance. Feng \textit{et al.} \cite{Feng11} considered the problem of efficient data aggregation in WSNs by putting in place amendment strategies in case of failures. Their solution needs local information to repair the aggregation tree and automatically reschedules nodes for interference free aggregation after the amendment. Cheng \textit{et al.} \cite{Cheng13} present STCDG, an efficient data gathering scheme based on matrix completion. STCDG takes advantage of the low-rank feature instead of sparsity, thereby avoiding the problem of having to be customized for specific sensor networks. They exploit the presence of the short-term stability feature in sensor data, which further narrows down the set of feasible readings and reduces the recovery errors significantly. Furthermore, STCDG avoids the optimization problem involving empty columns by first removing the empty columns and only recovering the non-empty columns, then filling the empty columns using an optimization technique based on temporal stability.

To preserve the overall energy in the network, sensor nodes are on a periodic schedule where they are switched on only when the sensing level is decreased. An optimal schedule needs to take nodes failure rate and the elapsed run-time into consideration. When the failure rate is small, wakening the nodes too often would only waste energy. As we go further in time, nodes start to exhaust their energy supply and this is when they start to fail. A combination a node failure rate and elapsed time would give us a better indication of the optimal nodes wakening schedule.\\
Maintaining the sensing level considerably reduces the amount of packet loss, yet it does not completely prevent its occurrence. A sudden node failure will result in the permanent loss of the held data packet, unless a redundancy mechanism is put in place. In the context of reducing energy consumption, the redundancy solution should be avoided and replaced by other solutions which do not include unnecessary packet transmission.\\ 

In this paper, we present a fault-tolerant data collection algorithm. This algorithm preserves energy consumption by only maintaining the necessary set of nodes in the active mode to ensure the minimum coverage level, while considering nodes failure rate. It is also able to recover data loss when a node fails before forwarding the data towards the base station. This algorithm is described in the next section.

\section{The proposed algorithm} 
\label{solution}

To cope with fault tolerance and data survivability, a fully distributed algorithm is presented and theoretically analyzed. Our algorithm seeks to cover data loss by maintaining a necessary set of working nodes and recovering failed ones when needed. We suppose that we are in the case of high density networks, and not all nodes participate in the network's service. Some nodes are in an idle state because their targets are actually covered by working sensors. We consider that these idle sensors wake up periodically to check for eventual node failures and therefore ensure their targets' coverage. In case of failures, they decide to switch to active mode and therefore initiate the recovery process to retrieve the data of the failed nodes. However, during the network's service, how can we handle the case where two (or more) sleeping nodes, would realize at the same time that the working neighbor is down?

Indeed, two neighboring sensor nodes may be elected at the same time step, and the recovery process of two neighboring nodes may be the same. This paper aims at filling this gap by proposing an efficient node failure recovery scheme in order to allow sensor networks to gracefully degrade in performance instead of failing unpredictably.

In the following, we first focus on the legitimate state formulation and next, we present the algorithm which consists in only three rules and give the correctness proofs.

\subsection{Problem formulization}

Let $G=(V,E)$ be the graph modeling the sensor network, with $|V|=n$ and $|E|=m$. We assume sensor node identifiers to be unique. We recall that sensor node identifier is unique if and only if $i.Id \neq j.Id$ holds for each $i,j \in V (i \neq j)$. A sensor node can be in one of these three states: \textit{failed, working,} or \textit{probing}. Every node $i$ in the network has to maintain the following data structure:

\begin{itemize}
\item [-]$D_i$: the sensed data by node $i$. Each time a node updates $D_i$, it sends/replicates the newly sensed data to/on its neighbors.
\item[-]$P_i$: the parity information on node $i$. It is the result of the combination of the replicated information of its neighbors.
\end{itemize}

We considered two different scenarios for the parity information. In the first scenario, there are no memory constraints. Each new data is saved on a different memory register, and we used the SUM function for data collection. In the contrary, all information must be saved on the same memory register when it comes to the second scenario. So, we used the XOR function to preserve memory space.

Let $T={t_1,t_2,..., t_k}$ be the set of monitoring targets to be covered and $S={1, 2,..., n}$ the set of sensor nodes. Each target in $T$ has to be covered by at least one sensor node in $S$. We call $\Gamma_u$ the set of neighbor- sensors of target $t_u , 1 \le u \le k$. Each neighbor-sensor $j \in \Gamma_u$ is capable of monitoring the target $t_u$, formally:

\begin{center}
$\forall j \in \Gamma_u: ds(t_u,j) \le R_s, \Gamma_u \subseteq S, t_u \in T,$
\end{center}

\noindent where $ds(t_u, j)$ denotes the distance between points $t_u$ and sensor $j$.

Let $N_i$ be the initial set of neighbors of node $i$ and $d_i=|N_i \ \backslash \ \Gamma_u(i) \ \backslash \ i|,$ the number of its working neighbors. As the number of failures goes up with time, we let $d^*_i$ be the dynamic number of alive neighbor nodes. We denote by $D^k$ the set of $d_i+1$ replicas of data $D_i$. Also, we denote by $s \ (D^k)$ the sensor node to which data-replica $D^k$ is assigned, for $1 \le k \le d_i+1$ and by $\hat{s} \ (D^k)$ the elected sensor node who recovers $D_i$ if node $i$ fails. The data are replicated on different nodes (space exclusion, see Lemma \ref{l1}) since the goal is to achieve data survivability even if some node failures occur in the network.

\noindent We say that a sensor node $i$ is independent if

$i.state = working \wedge (\forall j \in \Gamma_u(i))(j.state = sleeping \vee probing \vee failed)$

\noindent and that $i$ is dominated if

$(i.state = sleeping \vee probing) \wedge (\exists j \in \Gamma_u(i))(j.state = working)$

\noindent The legitimate state (let us denote it $\Sigma$) of the network is then expressed as follows:

$\forall i \in V: i.state = failed$ \\
\indent $\Rightarrow ((\exists \hat{s}, \hat{s}' \in \Gamma_u(i))(\hat{s}.state = \hat{s}'.state = working) \Rightarrow (\hat{s}(D^k_i) = \hat{s}'(D^k_i)))$\\

In other words, each data loss is recovered by at most one working sensor node.

\subsection{The algorithm}

When a sleeping node wakes up, it sends a \textit{probe-request} message to check if there exist working nodes in its vicinity. If no working nodes, it recovers the lost data of the failed node and starts to operate in the active mode; otherwise, it sleeps again. Nodes are initially in the sleeping mode. Each node sleeps for an exponentially distributed time generated according to a probability density function (PDF) $f(t) = \lambda e^{-\lambda t}$, where $\lambda$ is the probing rate of the sensor node and $t$ denotes its sleeping time duration.

Upon detecting an eventual failure, a probing node $i$ updates its actual probing rate $\lambda_i$ by taking into account the dynamic number of alive neighbors $d^*_i: \lambda^{new}_i \leftarrow \lambda_i \ . \ \frac{d_i}{d^*_i}$. Then, a new sleeping period is generated by using the new computed parameter $\lambda^{new}_i$ according to the PDF function: $f(t) = \lambda^{new} e^{-\lambda^{new}t}$.
The following notations are also given for the predicates of node $i$

\begin{itemize}
\item[-] $W(i)$: working neighbor: $\exists j \in \Gamma_u(i), i.state = working$
\item[-] $W^*(i)$: working neighbor with lower Id: $\exists j \in \Gamma_u(i), j.state = probing \wedge i.Id > j.Id$
\item[-] $F(i)$: failed neighbor: $\exists j \in \Gamma_u(i), j.state = failed$
\item[-] $P^*(i)$: probing neighbor with lower Id: $\exists j \in \Gamma_u(i), j.state = probing \wedge i.Id > j.Id$
\end{itemize}

\noindent The proposed algorithm uses the following three rules:\\
$r1$:

\begin{algorithmic}
\IF{$(i.state = probing \wedge (P^*(i) \vee W(i)))$}
\STATE{\IF{$P^*(i)$}
\STATE{$\lambda^{new}_i \leftarrow \lambda_i . \frac{d_i}{d^{*}_i}$}
\ENDIF}
\STATE{$i.state \leftarrow sleeping$}
\ENDIF
\end{algorithmic}

\begin{figure}[ht!]
	\begin{center}
		\includegraphics[width=0.5\textwidth]{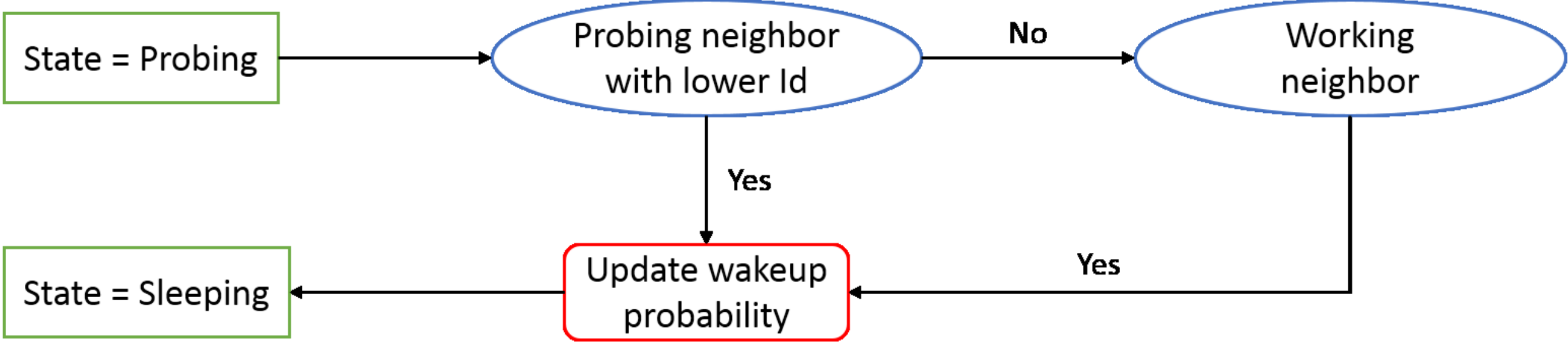}
		\label{r1}
		\caption{Algorithm rule 1.}
	\end{center}
\end{figure}

\indent
\newline
$r2$:

\begin{algorithmic}
\IF{$i.state = probing \wedge (\lnot W(i) \wedge \lnot P^*(i) \vee F(i))$}
\IF{$F(i)$}
\IF{ memory constraint} 
\STATE {$D_i \leftarrow P_z \underset{k \in N_z \backslash \Gamma_u(i), k\neq j} {\oplus} D_k$} \tt{(*$F(i)=F(z)=j$*)}
\ELSE 
\STATE {$D_i \leftarrow  D_k,   k \in N_z \backslash \Gamma_u(i), k\neq j$} \tt{(* $k$ is chosen randomly*)}
\ENDIF
\ENDIF
\STATE{$i.state \leftarrow working$}
\ENDIF
\end{algorithmic}

\begin{figure}[ht!]
	\begin{center}
		\includegraphics[width=0.5\textwidth]{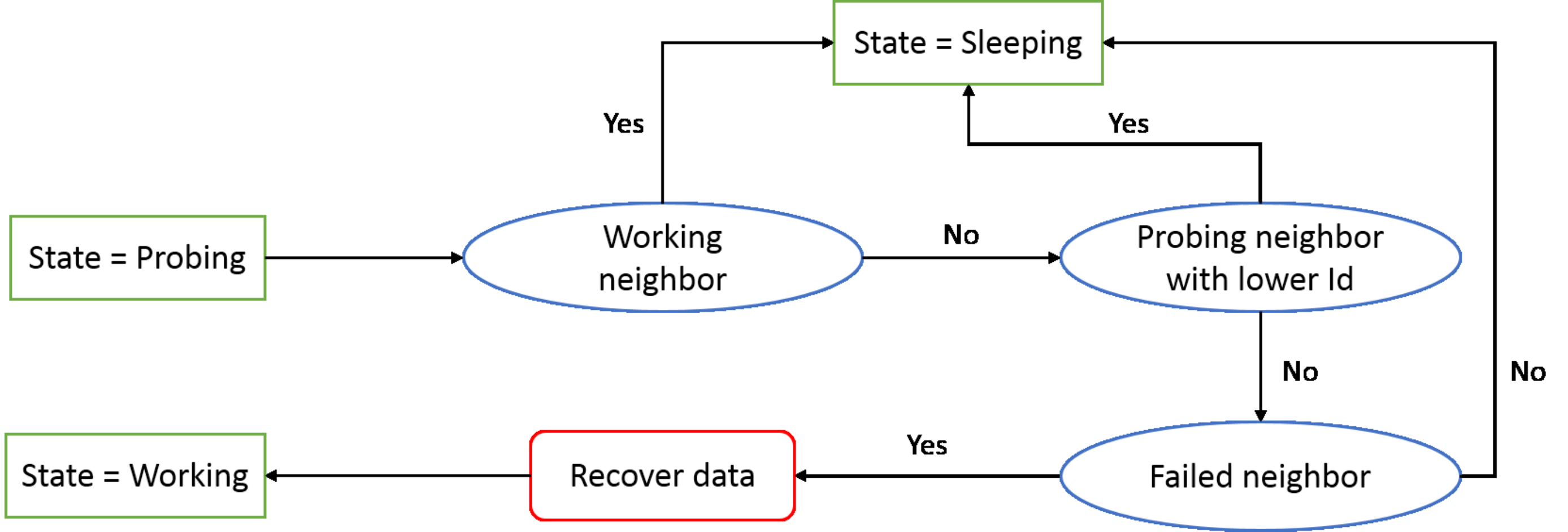}
		\label{r2}
		\caption{Algorithm rule 2.}
	\end{center}
\end{figure}

\indent
\newline
$r3$:

\begin{algorithmic}
\IF{$(i.state = working \wedge W^*(i))$}
\STATE{$i.state \leftarrow sleeping$}
\ENDIF
\end{algorithmic}

\begin{figure}[ht!]
	\begin{center}
		\includegraphics[width=0.5\textwidth]{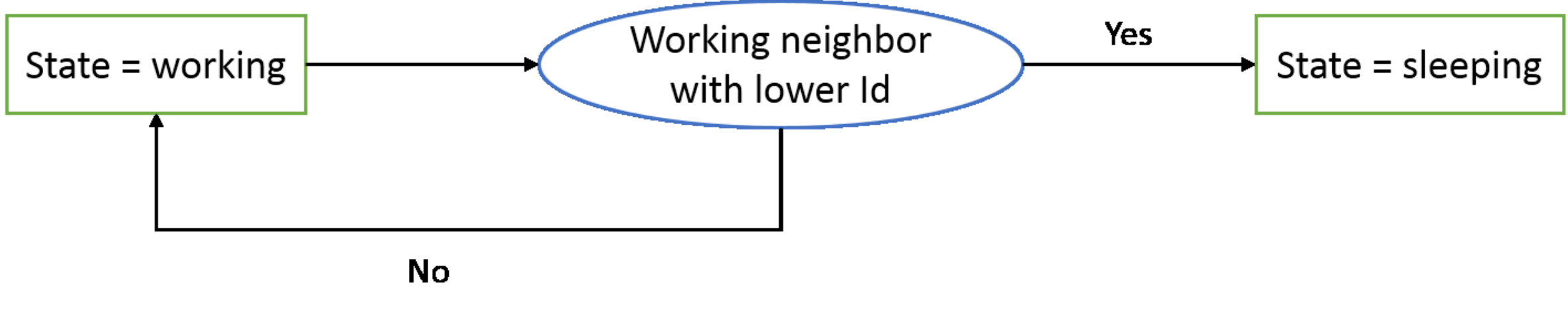}
		\label{r3}
		\caption{Algorithm rule 3.}
	\end{center}
\end{figure}

\subsection{Correctness proofs}

In this section, we detail properties of our fault tolerant algorithm, and express its validity/convergence. We assume that links are trustworthy/flawless and lossless.

\begin{Lemma}
\label{l1}
A sensed data $D_i$ is guaranteed to survive in the presence of $d_i$ permanent faults if and only if $s(D^{k}_i) \neq s(D^{k'}_i)$, for $1 \leq k$, $k' \leq d_i+1$.
\end{Lemma}

\begin{proof}
If $d_i$ nodes fail, then there is $s(D^{u}_i), 1 \leq u \leq d_i+1$ which did not fail, and therefore $D^{u}_i$ will be recovered successfully from $s(D^{u}_i)$ since there are $d_i+1$ copies of $D_i$ assigned to $d_i+1$ different nodes. However, if there is a sensor node $s(D^{k}_i), 1 \leq k \leq d_i+1$, such that $s(D^{k}_i) = s(t^{u}_i) = s^*$ and $s^*$ fails, then neither $D^{k}_i$ nor $D^{u}_i$ can be recovered successfully.
\end{proof}

\begin{Lemma}
\label{l2}
If at most $d_i$ neighbors crash down for any sensor node $i \in V$ in the network, then the algorithm is valid and resists to eventual node failures.
\end{Lemma}

\begin{proof}
The proposed algorithm is based on replication scheme with space exclusion. Thus, according to Lemma \ref{l1}, each data is replicated $d_i+1$ times onto $d_i+1$ distinct sensor nodes. We have at most $d_i$ node failures at the same time. So at least one copy of each data is recovered from a fault free node.
\end{proof}

\begin{Lemma}
\label{l3}
If a node changes to the working state by $r2$, then it remains in its state and will never execute a rule again until an eventual failure takes place.
\end{Lemma}

\begin{proof}
Let $i$ be a sensor node that executes $r2$. According to the preconditions of all rules, node $i$ can execute only rule $r3$ in the next round. However, in order to do so, one of its neighbors would have to switch to $working$ state following $r2$. This is impossible as long as node $i$ is in the $working$ state. Thus, node $i$ will never execute a rule again. If node $i$ is down, it remains in its state (fail-stop failure). 
\end{proof}

\begin{Lemma}
\label{l4}
If a sensor node is enabled by rule r2, then each one of its neighbors will execute at most one more rule until their next wake-up/probing, and this rule will be r1.
\end{Lemma}

\begin{proof}
Let $i$ be a node that executes $r2$. When node $i$ changes to $working$ state, all its neighbors are either in $sleeping$, $probing$, or $failed$ state. So, we have three possible scenarios: i) neighbors in $sleeping$ state: there is no conflict in this case. ii) neighbors with $probing$ state: these neighbors have a higher $Id$ than $i$. iii) failed neighbors will remain in their state until their recovery.
\end{proof}

\begin{Lemma}
\label{l5}
Every sensor node is either independent, dominant, or failed.
\end{Lemma}

\begin{proof}
From the point of view of node $i$, we have three scenarios:

\begin{itemize}
\item[-] if node $i$ is in the $working$ state and is not $independent$, then $i$ may execute rule $r3$.
\item[-] if node $i$ is in the $sleeping \vee probing$ state and is not $dominated$, then node $i$ may execute rule $r2$.
\item[-] if node $i$ is in the $failed$ state, then node $i$ will remain in its state until its recovery.
\end{itemize}

\end{proof}

\begin{Lemma}
\label{l6}
When a node is not failed $\vee$ sleeping, it can make at most two moves.
\end{Lemma}

\begin{proof}
By Lemma \ref{l3} and Lemma \ref{l4}, each rule can be executed at most once by a node. Hence, the only case a node makes two moves is when it executes $r3$ then $r2$ with a $working$ state.
\end{proof}

\begin{Theorem}
\label{t1}
With respect to the legitimate state $\Sigma$ of the network, the proposed algorithm converges within 2n moves.
\end{Theorem}

\begin{proof}
This follows from Lemma \ref{l1} to Lemma \ref{l6}.
\end{proof}

\subsection{Message complexity analysis}

In the following, we give an Upper-Bound of the actual number of probe/reply messages exchange during the network's lifetime task.

\begin{Theorem}
\label{t2}
The number of probe/reply messages involved by the algorithm is at most: $$O \left( n\ m \times \max_i \frac{t^{R_i}_i}{\Delta_i}\right), 1 \leq i \leq n $$
where, n is the number of nodes, m is the number of virtual communication links, $t^{R_i}_i$ is the reliable lifetime of node i and $\Delta_i$ is the smallest sleeping period time of node i.\\
This bound is attainable.
\end{Theorem}

\begin{proof}

(a)- The reliable lifetime $t^{R_i}_i$, of the node $i, 1\leq i\leq n$ for a specified reliability $R_i$, starting the mission at age $0$, is computed as follows:\\
$R_i = 1- F(t^{R_i}_i) = e^{-\lambda t^{R_i}_i} \Rightarrow ln\ R_i = -\lambda t^{R_i}_i \Rightarrow t^{R_i}_i = -\frac{1}{\lambda}\ ln\ R_i$\\
This is the lifetime during which the sensor node $i$ will be functioning successfully with a reliability of $R_i$.\\
According to node's sleeping periods subdivisions of the time, we have:\\
$0 = t_o < t_1 < t_2 < ... < t_k = t.$ Let $\Delta_p = [t_{p-1}, t_p[, 1 \leq k$ denote the $p^{th}$ sleeping period time. Since the number of failures goes up, the sleeping time period decreases with time. This implies that the probing process of node $i$ costs at most $O\left(\frac{t^{R_i}_i}{\Delta_i}\right), 1 \leq i \leq n, \Delta_i = min \Delta_p, 1 \leq p \leq k$. In addition, for each probing message issued from node $i$, we may have the corresponding reply messages from its working neighbors. This cost is at most $O(|N_i|)$. Therefore, from the point of view of node $i$, the number of probe/reply messages is at most $O \left(|N_i| \times \frac{t^{R_i}_i}{\Delta_i}\right)$.\\
Finally, summing up for the whole $n$ sensor nodes, the algorithm's message cost is at most $O\left(\Sigma^{n}_{i=1}|N_i|\times \frac{t^{R_i}_i}{\Delta_i}\right) \leq O\left(n\ m\times max_i \frac{t^{R_i}_i}{\Delta_i}\right), 1\leq i\leq n$ \\

(b)- To see that this bound is really attainable, consider a linear chain graph of only two sensor nodes $s_1$ and $s_2 (n=2)$. We need to orchestrate the involved communications between these nodes in time. Assume that $s_1$ is working and $s_2$ is the passive state. If $t^{R_1}_1 = t^{R_2}_2$ ($s_1$ and $s_2$ start functioning and fail at the same time), then the whole number of probe-message issued from $s_2$ is $\frac{t^{R_2}_2}{\Delta_2}$, where $\Delta_2$ is the constant sleeping time period of $s_2$. Since both $s_1$ and $s_2$ have the same life for which nodes will be functioning successfully, node $s_1$ will reply for each probing message issued from $s_2$. As a result, the whole number of involved probe-request/reply message before the failure of $s_1$ and $s_2$ is $n\ m\times \frac{max_i t^R_i}{min_i min_j \Delta_{j i}}= 2 \times \frac{t^{R_2}_2}{\Delta_2}$

\end{proof}

\section{Simulation results}
\label{results}

\begin{figure*}[ht!]
\begin{center}
\subfigure[Network's lifetime]{\includegraphics[width=6cm] {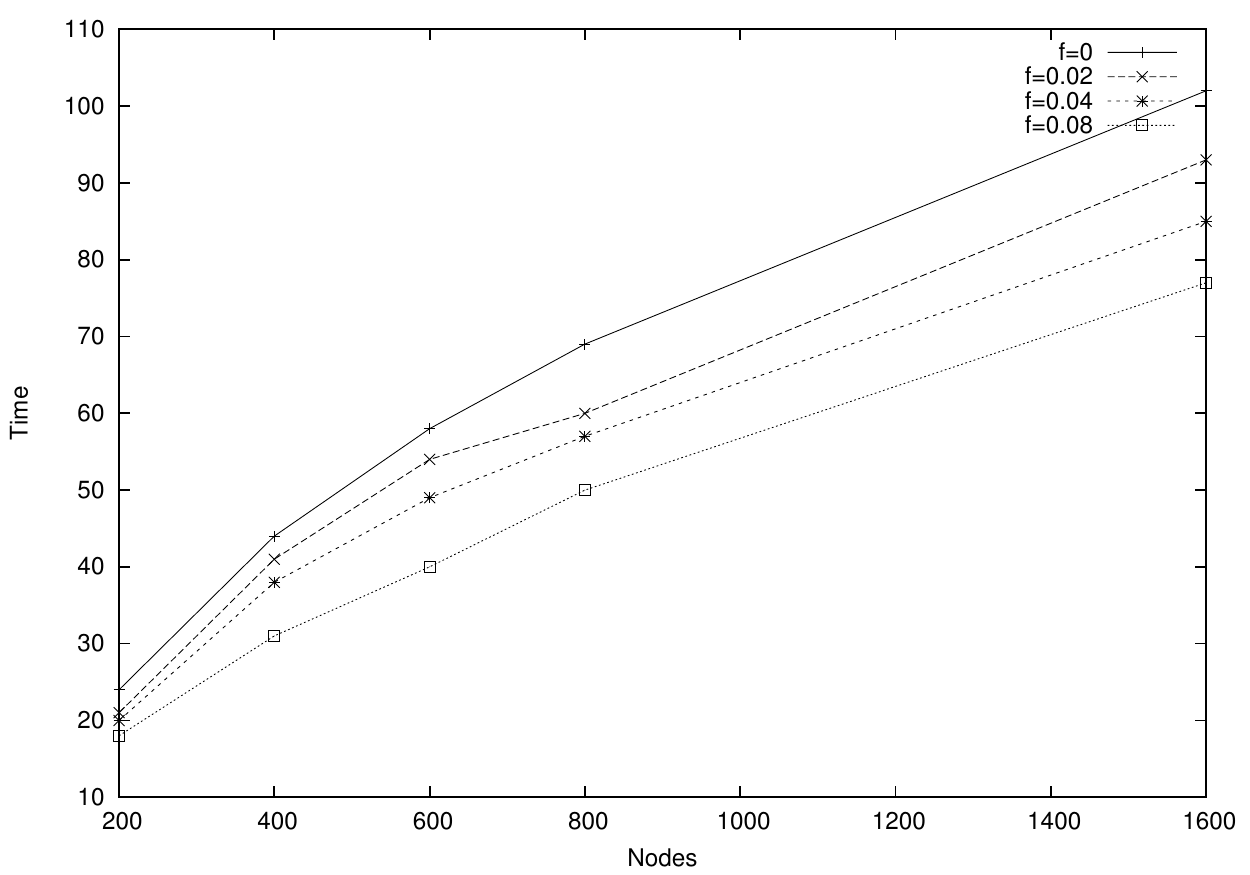}   
\label{fig1x:subfig1}}
~
\subfigure[Failures of the recovery process]{\includegraphics[width=6cm] {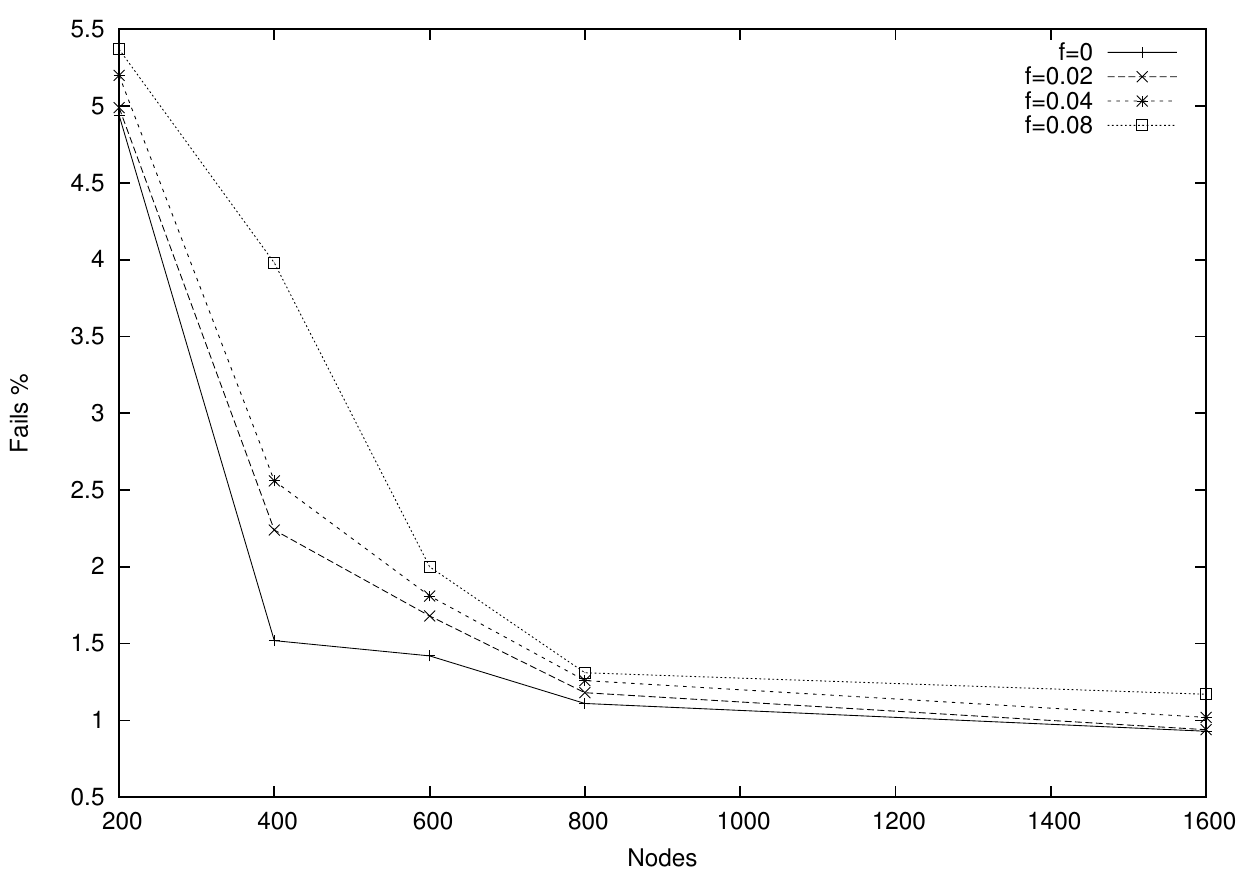}  
\label{fig1x:subfig2}}
~
\subfigure[Coverage rate]{\includegraphics[width=6cm] {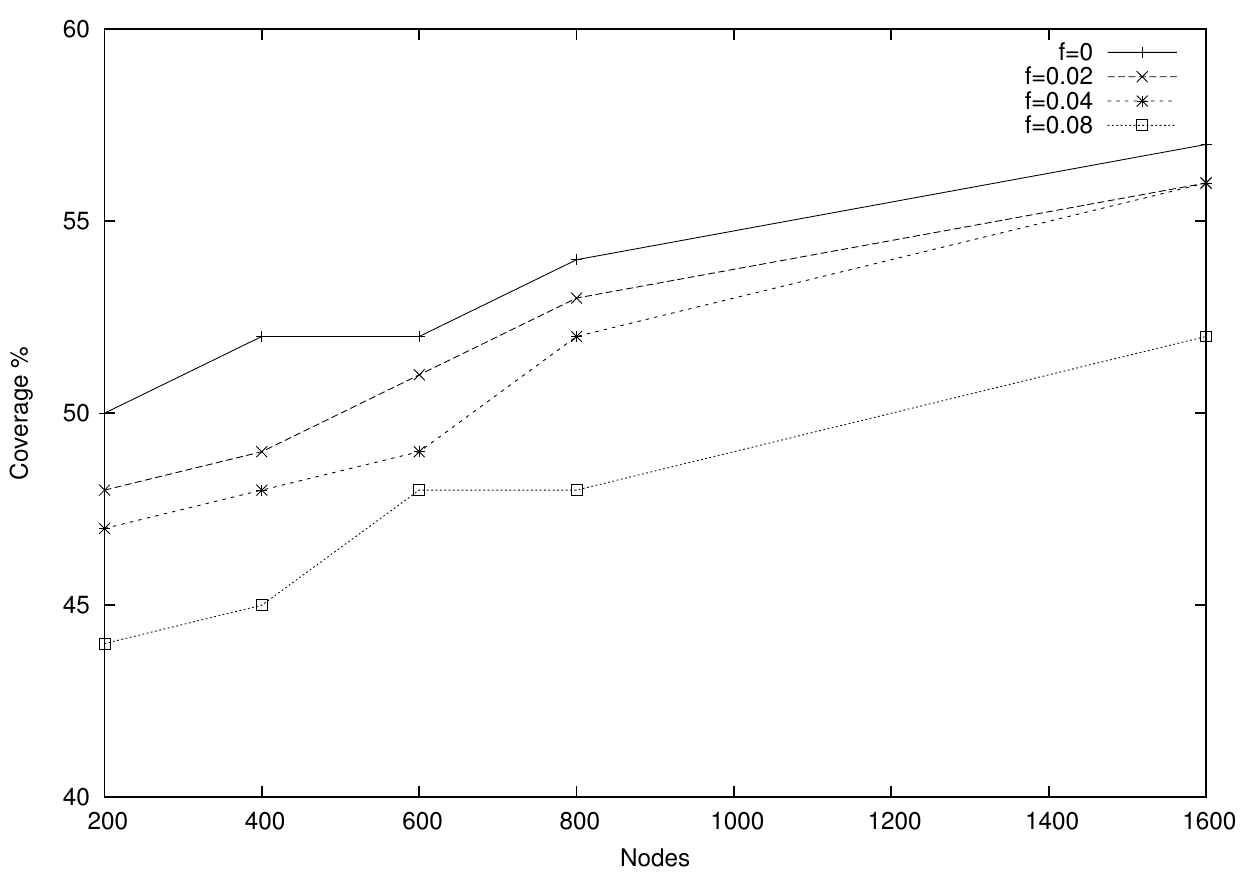}  
\label{fig1x:subfig3}}
~
\subfigure[Number of total messages in the network]{\includegraphics[width=6cm] {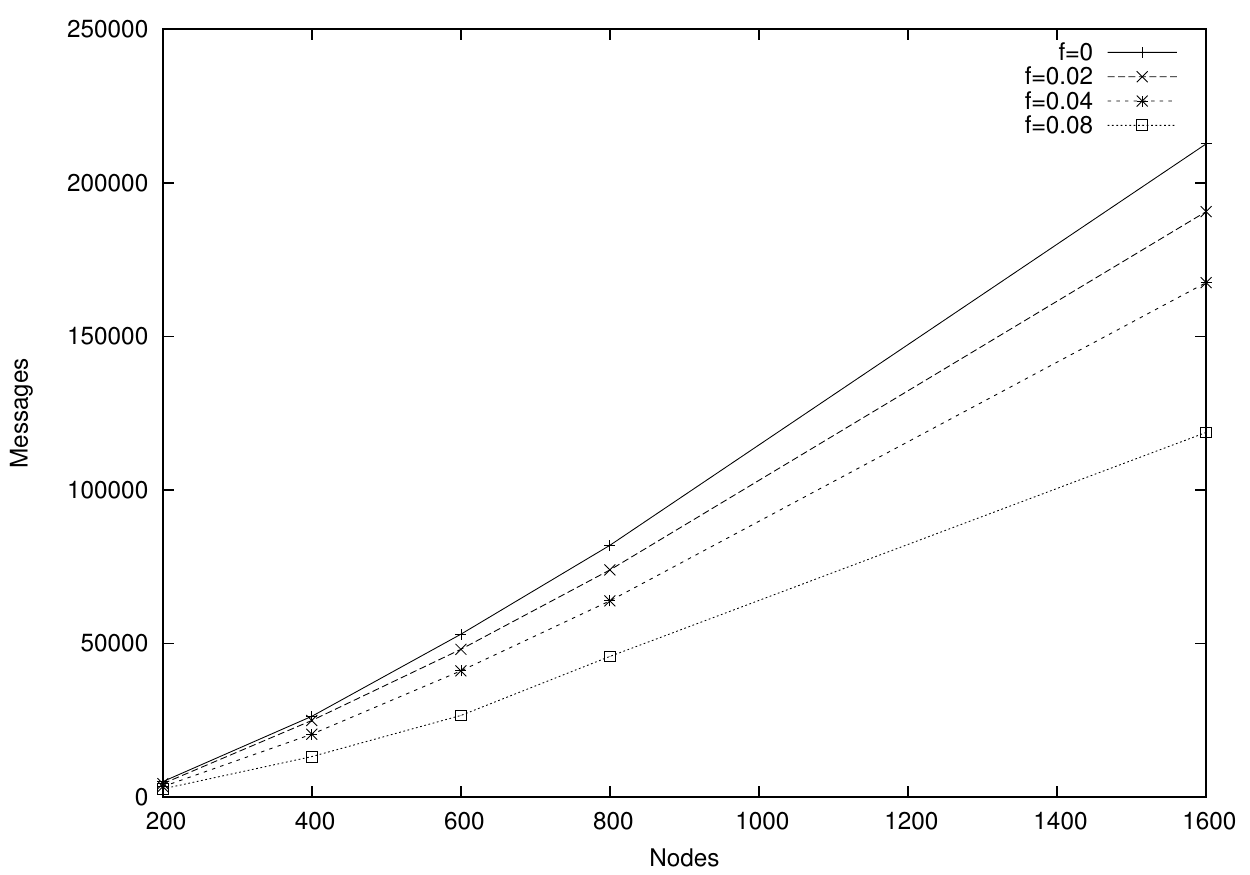}   
\label{fig1x:subfig4}}
~
\subfigure[Number of total wake-ups in the network]{\includegraphics[width=6cm] {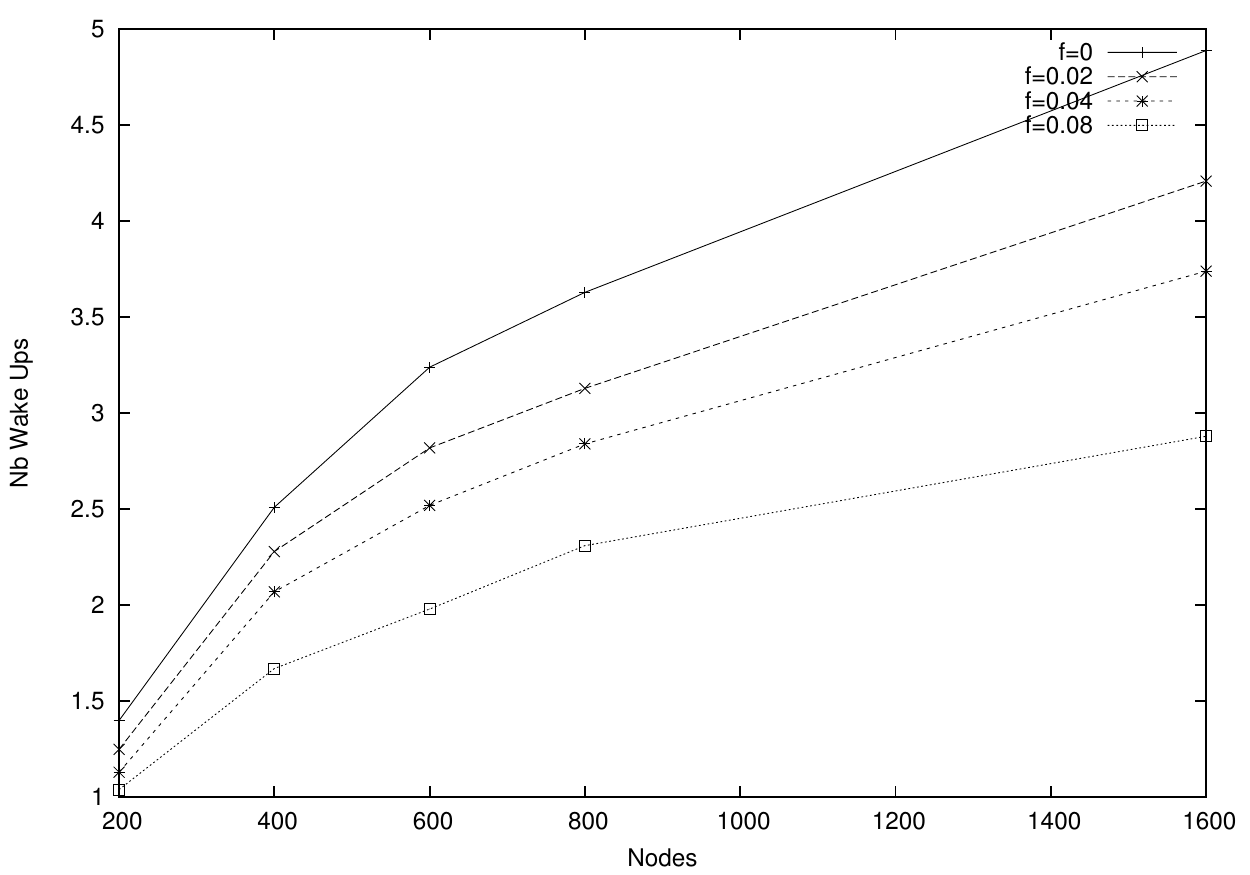}
\label{fig1x:subfig5}}
\caption{Performance evaluation with average (1x) wake up rate.}
\label{Fig1x}
\end{center}
\end{figure*} 

\begin{figure*}[ht!]
\begin{center}
\subfigure[Network's lifetime]{\includegraphics[width=6cm] {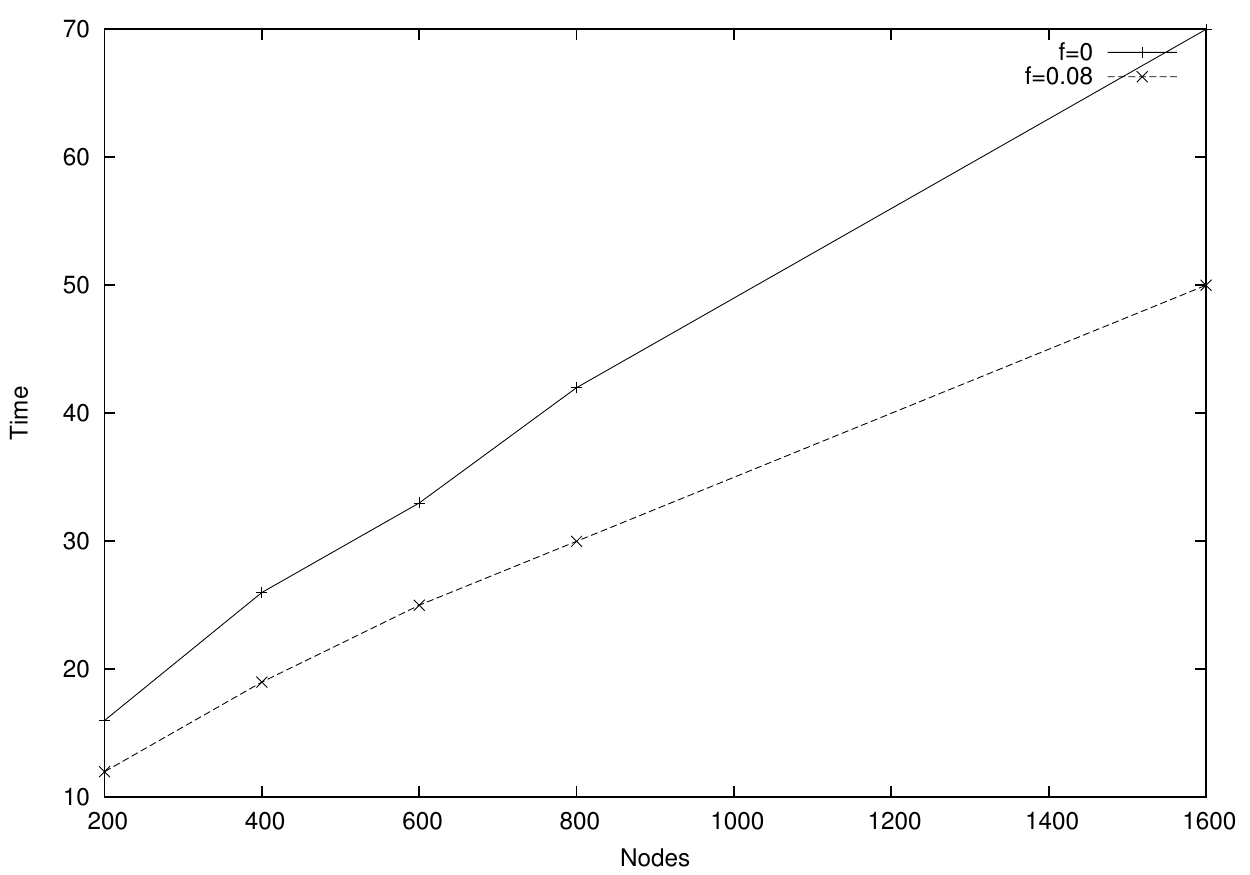}   
\label{fig4x:subfig1}}
\subfigure[Failures of the recovery process]{\includegraphics[width=6cm] {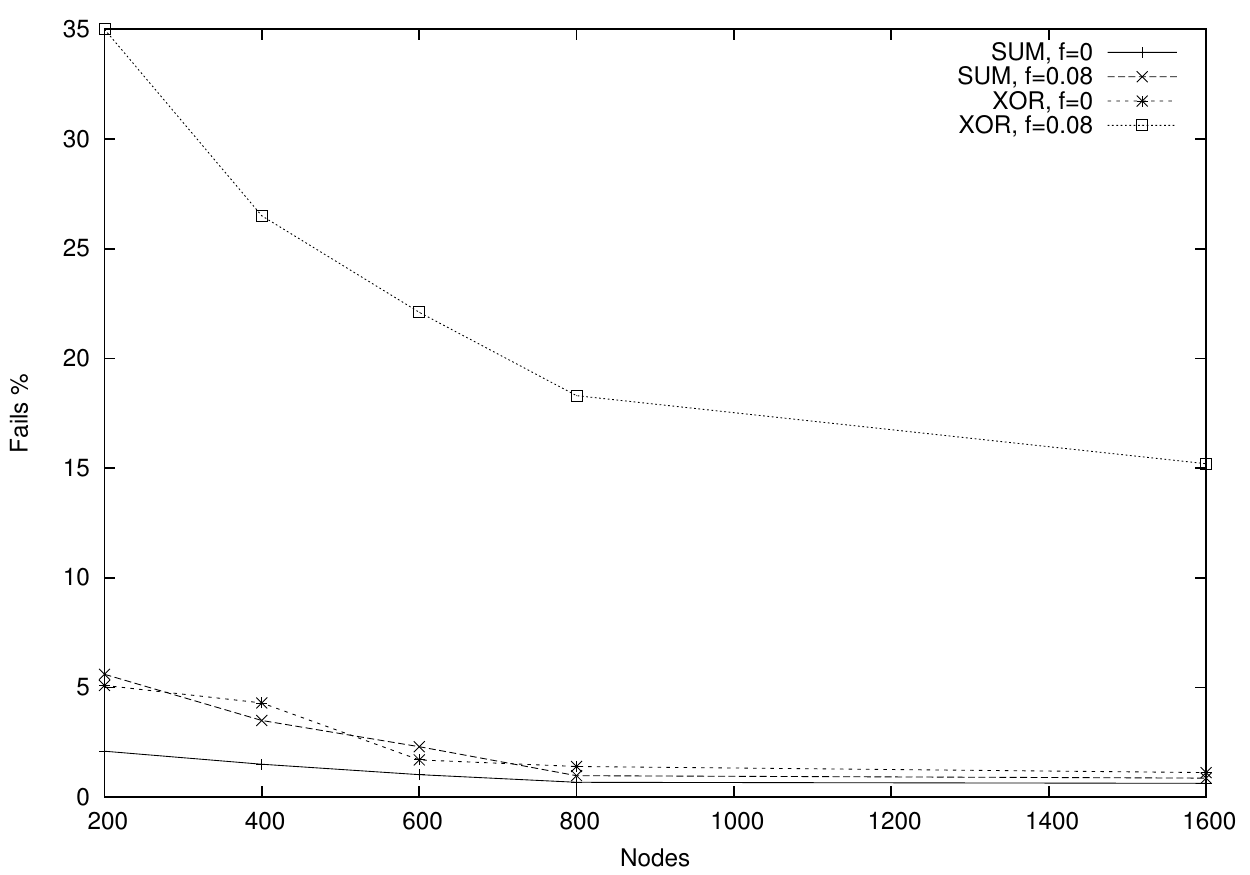}  
\label{fig4x:subfig2}}
\subfigure[Coverage rate]{\includegraphics[width=6cm] {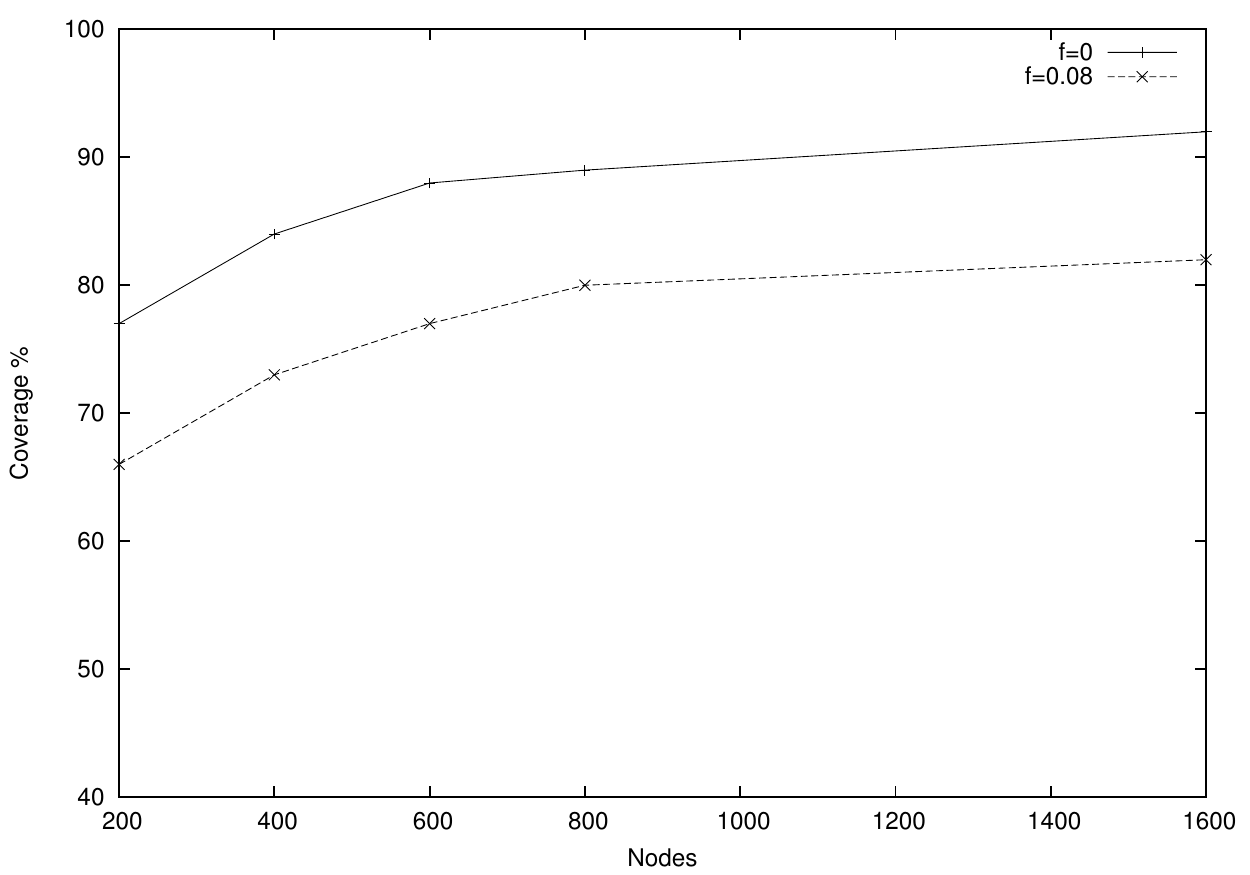}  
\label{fig4x:subfig3}}
\subfigure[Number of total messages in the network]{\includegraphics[width=6cm] {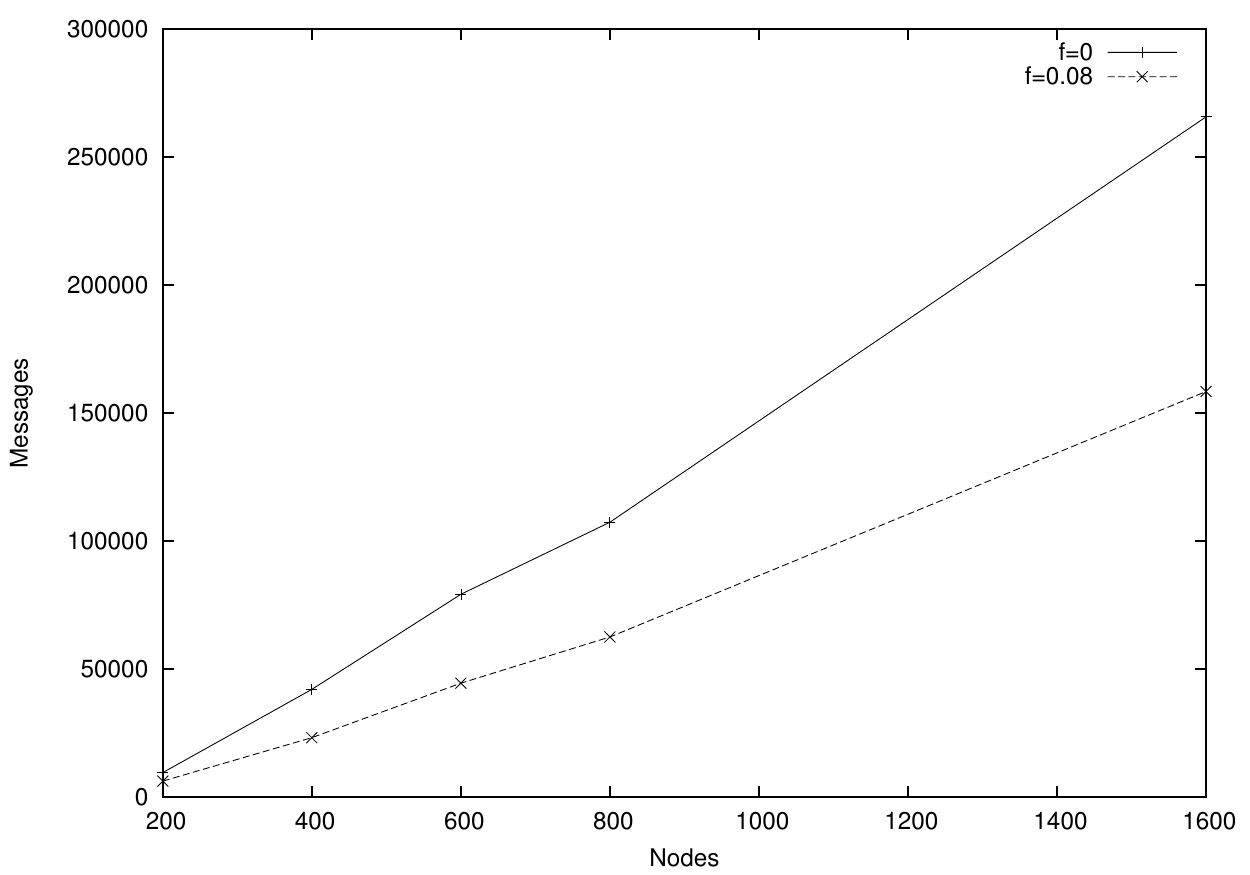}   
\label{fig4x:subfig4}}
\subfigure[Number of total wake-ups in the network]{\includegraphics[width=6cm] {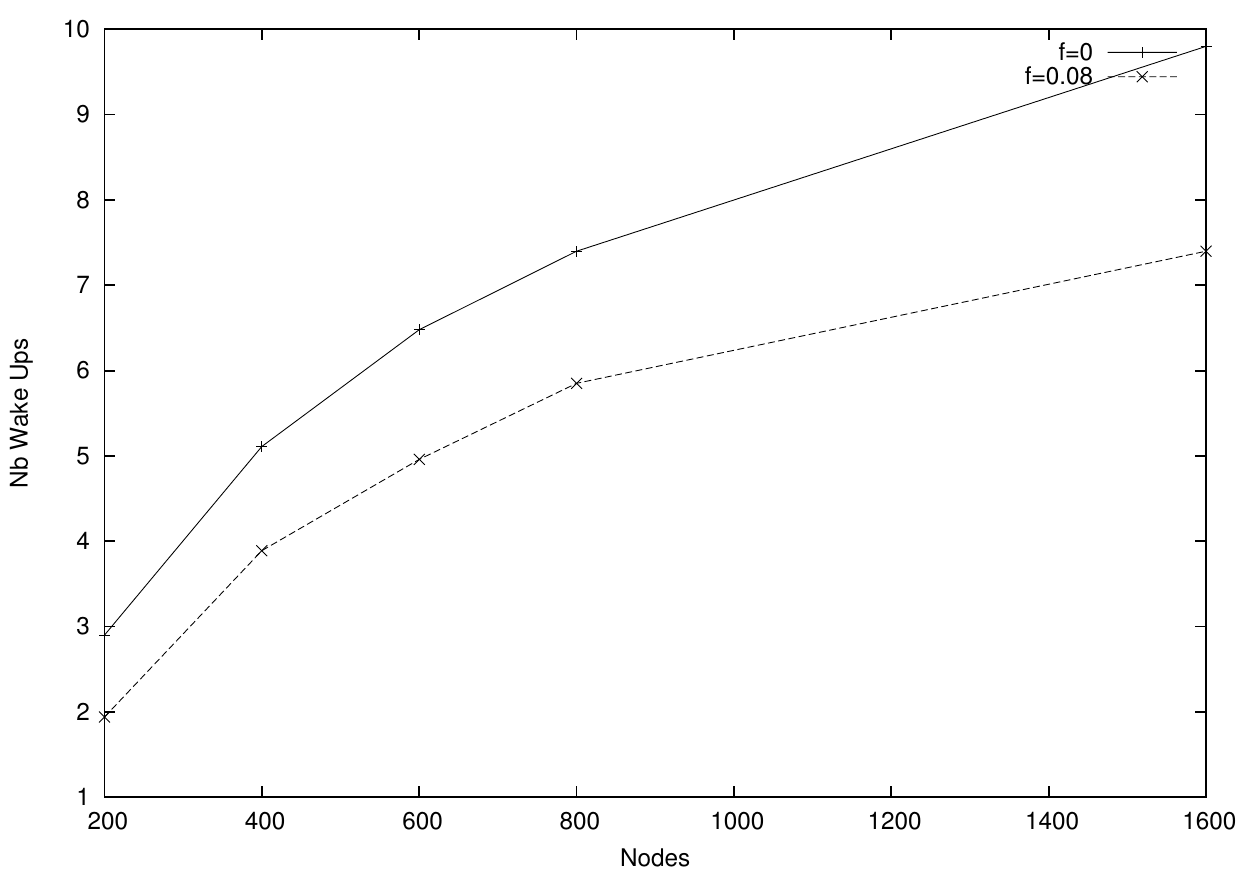}
\label{fig4x:subfig5}}
\caption{Performance evaluation with high (4x) wake up rate.}
\label{Fig4x}
\end{center}
\end{figure*}

In this section, we discuss some results through simulations. We consider a flat grid topology of $10$ by $10$ \textit{i.e.} 100 monitoring zones. We vary the number of sensors between $200$ and $1600$ nodes. Since sensors are uniformly distributed in the monitoring area, the density of sensors at each zone varies between $2$ and $16$. 

The performance evaluation considers five aspects: (i) Network lifetime evolution; (ii) Failure rate: that is the ratio of information recovery attempts that did not succeed; (iii) Effective monitoring time: this measure is related to the time between the death of the active node in a monitoring zone and its replacement; it is expressed in (\%); (iv) Total number of messages; and (v) Number of awakenings per inactive sensors.

Our simulations are performed in two different settings: the first setting sets a low wake-up rate but enough to keep the monitoring time ratio higher than 50\% in almost all configurations (see Figure \ref{Fig1x}); while the second setting considers a wake-up rate four times higher than the rate in the previous setting. This allowed us to reach monitoring time ratios up to 90\% (see Figure \ref{Fig4x}).

These two settings were put in place in order to test and compare two different solutions for data collection. In the first scenario, we assume that there is no constraint related to memory capacity of sensor nodes. Therefore, each sensor is able to save data received from all nodes in its neighborhood on a different memory register (this method is called SUM). As for the second scenario, we aim for preserving the memory space. Thus, we suppose that each node has one memory register that is available to save all information received from all its neighboring nodes; every new data packet is added in the register using the function XOR.
We can notice that the second method for data saving (XOR) is highly sensitive to any neighboring node failure. In fact, the failure of a neighbor induces a corruption of the calculation of the data needed for the recovery process. For this reason, the coverage rate needs to be high enough for this solution to work. Consequently, the XOR method is only implemented in the 4x version of our simulation settings.

\subsection{Wake-up rate = 1x}

In this section, only the SUM method is implemented. The nodes in the network are designed to fail randomly. This failure rate varies from $0\%$ to $8\%$ by a pitch of $2\%$.

In Figure \ref{fig1x:subfig1}, we can observe that the network's lifetime increases in an almost linear manner. For a small network, the number of times a sensor node receives a wake-up message is higher comparing to the same number when the network is larger. Therefore, the more nodes are participating in the coverage process, the less energy is consumed per node and the overall energy in the network is preserved.
The overall energy level has an impact on the recovery process. In fact, when the energy level start to go down, the number of nodes able to cover a given area is reduced. The remaining nodes will receive an increasing number of wake-up messages, and when they fail the number of replacements is continuously decreased. Eventually, some zones will no longer be covered. Thus, data recovery rate increases when the network is larger. Even though failures are less impacting for a dense network, the failure rate is low even for a small network (see Figure \ref{fig1x:subfig2}). 
The coverage rate (illustrated in Figure \ref{fig1x:subfig3}) is more successful when the density of the network grows. Nevertheless, it remains between the values of $45\%$ and $55\%$ depending on the settings. Indeed, as discussed above, the number of wake-up messages in the network is higher with the growth of nodes number and therefore more energy is dissipated. On the other hand, when the number of nodes is small, even though the number of wake-up messages is reduced, nodes fail faster as the number of node replacements is small. Consequently, there is no huge difference in the coverage rate. Still, a dense network guarantees a better coverage rate.
The total number of exchanged messages in the network will only grow with the increased number of nodes in the network as shown in Figure \ref{fig1x:subfig4}. In fact, each node will copy its sensed data onto each one of its neighbors. So when the number of nodes increases, the number of messages increases accordingly.
The number of wake-ups per node illustrated in Figure \ref{fig1x:subfig5} follows a logarithmic form. Sensor nodes periodically wake up to verify if there zone is being covered by an active node. The wake-up rate follows a probability function that is updated considering node failure. So, the number of these messages highly depends on the number of nodes failure. When the probability of node failure increases, the number of nodes in the network is decreased and thus the total number of wake-ups.

\subsection{Wake-up rate = 4x}

In this section, both of SUM and XOR methods are tested. Since all the curves are similar, except for failure of the recovery process, only the figure corresponding to the latter illustrates the comparison between both methods. Nodes failure rate are fixed to $0\%$ and $8\%$ only (the two extreme cases from the previous configuration).

Comparing to the previous configuration, the network's overall lifetime has decreased. Considering that the wake-up rate here is $4$ times more frequent, it is normal that network consumes more energy in this setting (see Figure \ref{fig4x:subfig1}). Nevertheless, the different zones coverage rate illustrated in Figure \ref{fig4x:subfig3} was considerably and understandably improved. The failure of the recovery process in Figure \ref{fig4x:subfig2} remains very low with the absence of memory constraints, and even lower comparing to the previous configuration. In the contrary, the XOR function appears to be highly sensitive to node failure. When the failure rate reaches $8\%$, the recovery failure jumps by $30\%$ for a small network and $15\%$ when the network is dense.
The total number of exchanged messages is considerably higher than the number in the previous configuration, and this is due to the increased number of wake-up messages. The algorithm also improves the overall energy consumption by only maintaining a necessary set of nodes in the active mode. the rest of the node wake up randomly to check their area and ensure that coverage is performed. This random function is optimized by updating it accordingly to the nodes failure rate.

\section{Conclusion}

In this paper, we proposed a fully distributed algorithm that seeks to cover data loss by maintaining a necessary set of working nodes and recovering failed ones when needed. Each sensor node copies its data onto neighbors using two different assumptions: (i) in the first one, we suggest that there is no memory constraint and each new information is copied on a different register, and (ii) in the second one we put in place a memory constraint and use the XOR function to add a new data to the common memory register for all data. We also tested two different configurations, where the wake-up rate is $4$ times more frequent from one configuration to the other. The performed simulations showed that a more frequent wake-up rate helps improve the quality of the recovery process. Even though the absence of memory constraints facilitates the recovery process, this rate was maintained below $35\%$ for a small network and around $15\%$ for a dense one even in the presence of memory constraints. This algorithm also helps preserving the energy in the network by only maintaining a necessary set of sensor nodes in the active mode. The rest of the nodes wake up randomly to ensure that their area is covered by a sensor node. This random function is optimized by updating it according to the nodes failure rate.

\nocite{*}

\bibliographystyle{compj}
\bibliography{biblio}

\end{document}